\documentclass[12pt]{IEEEtran}
\IEEEoverridecommandlockouts
\usepackage{graphicx}
\usepackage{amsmath}
\usepackage{bigints}
\usepackage{booktabs}
\usepackage[inline]{enumitem}
\usepackage[export]{adjustbox}
\usepackage{wrapfig}
\usepackage{subfigure}
\usepackage{comment}
\usepackage{hyperref}

\usepackage{amssymb,amsfonts} 
\usepackage{amsthm} 
\usepackage{commath} 
\usepackage{algorithmic}
\usepackage{mathtools, nccmath}

\usepackage{array}
\usepackage{textcomp}
\usepackage{xcolor}
\usepackage[final]{pdfpages}
\PassOptionsToPackage{bookmarks=false}{hyperref}
\def\BibTeX{{\rm B\kern-.05em{\sc i\kern-.025em b}\kern-.08em
    T\kern-.1667em\lower.7ex\hbox{E}\kern-.125emX}}

\newcommand{\sint}{\mathop{\mathpalette\dosint\relax}\!\int}
\newcommand{\dosint}[2]{%
  \ifx#1\displaystyle
    \displaysint
  \else
    \normalsint{#1}%
  \fi
}
\newcommand{\displaysint}{\displaystyle\mathsf{s}\mkern-18mu}
\newcommand{\normalsint}[1]{%
  \smallers{#1}\ifx#1\textstyle\mkern-9mu\else\mkern-8.2mu\fi
}
\newcommand{\smallers}[1]{%
  \vcenter{\hbox{$\ifx#1\textstyle\scriptstyle\else\scriptscriptstyle\fi\mathsf{s}$}}%
}

\theoremstyle{definition}
\newtheorem{theorem}{Theorem}

\theoremstyle{definition}
\newtheorem{definition}{Definition}

\theoremstyle{definition}

\theoremstyle{definition}

\theoremstyle{definition}
\newtheorem{lemma}{Lemma}

\theoremstyle{definition}

\begin{document}
\title{A Note on Comparator-Overdrive-Delay Conditioning
for Current\nobreakdash-Mode Control}
\onecolumn

\author{Xiaofan~Cui,~Guanyu~Qian,~and~Al-Thaddeus~Avestruz%
\thanks{1) Department of Electrical and Computer Engineering, University of California, Los Angeles, CA 90095, USA 
(e-mail: \href{mailto:gyqian@ucla.edu}{gyqian@ucla.edu}; \href{mailto:cuixf@seas.ucla.edu}{cuixf@seas.ucla.edu}).}%
\thanks{2) Department of Aerospace Engineering, University of Michigan, Ann Arbor, MI 48109, USA 
(e-mail: \href{mailto:avestruz@umich.edu}{avestruz@umich.edu}).}%
}

\IEEEoverridecommandlockouts
\IEEEpubid{\makebox[\columnwidth]{\hfill} \hspace{\columnsep}\makebox[\columnwidth]{ }}
\maketitle
\IEEEpubidadjcol

\begin{abstract}
\!\!Comparator-overdrive-delay conditioning is a new control conditioning approach for high-frequency current\nobreakdash-mode control.
No existing literature rigorously studies the effect of the comparator overdrive delay on the current\nobreakdash-mode control.
The results in this paper provide insights on the mechanism of comparator-overdrive-delay conditioning.
\end{abstract}

\section{Introduction}
High-frequency current\nobreakdash-mode control is one of the most popular
controller strategies for power converters \cite{Cui2018a,Fernandes2016}.
In high-frequency dc--dc converters, it is not accurate enough to simply model the comparator as an ideal logic device.
In a non-ideal comparator, a finite propagation delay exists between the input and output, and this delay depends on the \emph{input overdrive}—the instantaneous voltage difference between the input terminals.
This phenomenon is known as the \emph{comparator overdrive delay} (COD)~\cite{compoverdelayadi}.
We emphasize that this nonideality can be deliberately leveraged to attenuate interference in the current-sensing signal, thereby improving stability and transient accuracy.

To formalize this observation, Section~II establishes the theoretical framework of comparator-overdrive-delay conditioning in current-mode control loops, followed by Section~III, which quantifies the effect of interference. Section~IV then presents a practical case study that validates the proposed comparator-overdrive-delay conditioning approach.

\section{Theory}
\begin{theorem} \label{theorem:gmaxcontinuity}
Given a constant off\nobreakdash-time current control loop with comparator overdrive delay, if the input is a ramp with slope $m_1$ and interference function $w(t)$, the condition to guarantee the continuous static current mapping is
\begin{align}
  V_{\text{trig}}\tau \ge m_1 K_3\!\left(\frac{|W(\omega)|}{m_1}\right).
\end{align}
where $W(\omega)$ is the Fourier transform of the interference function $w(t)$.
\begin{equation}
W(\omega) \;\triangleq\; \int_{-\infty}^{\infty} w(t)\, e^{-j\omega t}\, dt
\end{equation}

\end{theorem}
\begin{proof}
We start with defining the saturating integral operator
\begin{definition} \label{def:sat_sum}
Given a function $f(x)$ that is continuous on the interval $[a,b]$, we divide the interval into $n$ subintervals of equal width $\Delta x$ and from each interval choose a point, $x_i^\ast $. The \emph{saturating partial sum} of this sequence is
\begin{align}
    \widehat{S}_{n+1} \triangleq u\left( \widehat{S}_{n} + f(x^\ast _i[n+1]) \Delta x \right),
\end{align}
where $u(x)$ is a saturating function where $u(x) = 0$ if $x < 0$ and $u(x) = x$ if $x \ge 0$.
\end{definition}
\begin{definition} \label{def:sat_int}
The \emph{saturating integral operator} is defined as
\begin{align}
    \sint_{a}^{b} f(x) \,dx \triangleq \displaystyle\lim_{n \rightarrow \infty}  \widehat{S}_{n}.
\end{align}
\end{definition}
We consider the following implicit mapping, which defines a function $\mathcal{M} \!: b \! \mapsto \! \theta $
\begin{align}
    \sint_{-\infty}^{\theta} \, \left( x + \frac{1}{2\pi} \int_{-\infty}^{+\infty} 
    | W(w) | e^{j(\omega x + \varphi(\omega))}\,d\,\omega - b \right) \, dx = k,
\end{align}
where $k$ paramaterizes the transconductance, effective capacitor, and threshold voltage as $k \triangleq V_{\text{trig}}C_{\text{eff}}/G $.
We define functions $K_1$ and $K_2$ as
\begin{align}
     & K_1(|W(\omega)|) \triangleq \underset{\substack{b \in \mathcal{R} \\ 0 \le \varphi(\omega) \le 2\pi} }{\text{sup}}\, K_2(b, \varphi(\omega)), \\
     & K_2(b,\varphi(\omega))  \triangleq  \sint_{-\infty}^{\psi_h} \,  \left( x + \frac{1}{2\pi} \int_{-\infty}^{+\infty}
    | W(\omega) | e^{j(\omega x + \varphi(\omega))}\,d\,\omega\;-\;b \right) \, dx ,
\end{align}
where $\psi_l$ and $\psi_h$, as the function of $b$ and $\varphi(\omega)$, are the lowest and largest solution of the equation
\begin{align}
    x + \frac{1}{2\pi} \int_{-\infty}^{+\infty}
    | W(\omega) | e^{j(\omega x + \varphi(\omega))}\,d\,\omega\;-\;b = 0.
\end{align}
It is not very clear if $K_1$ exists or not. We further introduce $K_3$ and $K_4$ to analyze.
We define functions $K_3$ and $K_4$ as
\begin{align}
     & K_3(|W(\omega)|) \triangleq \underset{\substack{- A \le b^{'} \le A \\ 0 \le \varphi^{'}(\omega) \le 2\pi} }{\text{max}}\, K_4(b^{'},\varphi^{'}(\omega)), \\
     & K_4 (b^{'},\varphi^{'}(\omega)) 
     \triangleq \sint_{-\infty}^{\psi^{'}_h} \,  \left( x + \frac{1}{2\pi} \int_{-\infty}^{+\infty}
    | W(\omega) | e^{j(\omega x + \varphi^{'}(\omega))}\,d\,\omega\;-\;b^{'} \right) \, dx,
\end{align}
where $\psi^{'}_l$ and  $\psi^{'}_h$, as the function of $b^{'}$ and $\varphi^{'}(\omega)$, are the lowest and largest solution of the equation 
\begin{align}
    x + \frac{1}{2\pi} \int_{-\infty}^{+\infty}
    | W(\omega) | e^{j(\omega x + \varphi^{'}(\omega))}\,d\,\omega\;-\;b^{'}= 0.
\end{align}
\begin{lemma}
We show that given $b$ and $\varphi(\omega)$, we can always find $b^{'}$ and $\varphi^{'}(\omega)$ such that \mbox{$K_4(b^{'},\varphi^{'}(\omega)) = K_2(b,\varphi(\omega))$}, and vice versa.
\end{lemma}
\begin{proof}
We define
\begin{align}
    \varphi^{'}(\omega) & = \varphi(\omega) + \omega \psi_l, \\
     b^{'} &= b -\psi_l.
\end{align} 
We first show the corresponding $\psi^{'}_h$ and $\psi_h$ as well as $\psi_h$ and $\psi_l$ satisfy the following relationships
\begin{align} 
    \psi^{'}_h & = \psi_h - \psi_l,\\
    \psi^{'}_l & = 0.
\end{align}
The relationship can be verified as
\begin{align}
    \psi_h + \frac{1}{2\pi} \int_{-\infty}^{+\infty}
    | W(\omega) | e^{j(\omega \psi_h + \varphi(\omega))}\,d\,\omega \;-\; b = 0,
\end{align}
\begin{align}
     \psi_h - \psi_l + \frac{1}{2\pi} \int_{-\infty}^{+\infty}
    | W(\omega) | e^{j(\omega \psi_h + \varphi^{'}(\omega) - \omega \psi_l)}\,d\,\omega
    \;-\; (b - \psi_l) = 0.
\end{align}
We validate that $\psi^{'}_h$ is the solution
\begin{align}
    \psi^{'}_h + \frac{1}{2\pi} \int_{-\infty}^{+\infty}
    | W(\omega) | e^{j(\omega \psi^{'}_h + \varphi^{'}(\omega))}\,d\,\omega \;-\; b^{'} = 0.
\end{align}
Similarly, we can validate that zero is the solution
\begin{align}
    K_4(b^{'},\varphi^{'}(\omega)) =& \sint_{0}^{\psi^{'}_h} \left( x + \frac{1}{2\pi} \int_{-\infty}^{+\infty}
    | W(\omega) | e^{j(\omega x + \varphi^{'}(\omega))}\,d\,\omega \;-\; b^{'} \right) dx \nonumber \\
    = &\sint_{0}^{\psi^{'}_h} \Bigg( x + \psi_l  + \frac{1}{2\pi} \int_{-\infty}^{+\infty}
    | W(\omega) | e^{j(\omega x + \varphi(\omega) + \omega \psi_l)}\,d\,\omega \;-\; b \Bigg)\,dx \nonumber \\
     =& \sint_{\psi_l}^{\psi_h} \left( x + \frac{1}{2\pi} \int_{-\infty}^{+\infty}
    | W(\omega) | e^{j(\omega x + \varphi(\omega) )}\,d\,\omega \;-\; b \right) dx \nonumber \\
    = & K_2(b,\varphi(\omega)).
\end{align}
The opposite direction is similar so we omit the proof.
\end{proof}
Therefore, $K_3$, defined as the maximum of $K_4$, should be equal to $K_1$, defined as the maximum of $K_2$,
\begin{align}
    K_1\!\left(\left|W(\omega)\right|\right) = K_3\!\left(\left|W(\omega)\right|\right).
\end{align}
$K_3(|W(\omega)|)$ exists from the Weierstrass theorem, hence $K_1(|W(\omega)|)$ exists.

The following Lemma\,\ref{theorem:contM} provides the 
condition that guarantees the static mapping to be continuous.
\begin{lemma} \label{theorem:contM}
For all $k>K_1(|W(\omega)|)$, function $\theta = \mathcal{M}(b)$ is continuous. 
\end{lemma}
\begin{proof}
From the subadditivity,
\begin{align}
    \bigg|\sint_{a}^{b} f(x) \,dx - \sint_{a}^{b} g(x) \,dx \bigg| \le \sint_{a}^{b} \bigg| f(x) - g(x)\bigg| \,dx.
\end{align}
Then we prove given $b$ which approaches $b_0$, $\theta$ approaches $\theta_0$.
$b$ and $\theta$ satisfies the integral-threshold constraint
\begin{align}  \label{eqn:b_theta_int}
    \sint_{-\infty}^{\theta} \, ( x + w(x) - b) \, dx = k.
\end{align}
$b_0$ and $\theta_0$ satisfies the integral-threshold constraint
\begin{align} \label{eqn:b0_theta0_int}
    \sint_{-\infty}^{\theta_0} \, ( x + w(x) - b_0) \, dx = k.
\end{align}
From (\ref{eqn:b_theta_int}) and (\ref{eqn:b0_theta0_int}),
\begin{align}
      & \sint_{-\infty}^{\theta_0} \, ( x + w(x) - b_0) \, dx - \sint_{-\infty}^{\theta_0} \, ( x + w(x) - b) \, dx \nonumber
\end{align}
\begin{align}
      & =\sint_{-\infty}^{\theta} \, ( x + w(x) - b) \, dx - \sint_{-\infty}^{\theta_0} \, ( x + w(x) - b) \, dx.
\end{align}
Denote the right hand side of equation as RHS. Because the $k>k_{\text{min}}$, $\theta > \theta(b)$ and $\theta_0 > \theta(b_0)$,
\begin{align}
    | \text{RHS} | & = \bigg|\int_{\theta_0}^{\theta} \, ( x + w(x) - b) \, dx \bigg| \;\ge\; \mu_1| \theta - \theta_0|,
\end{align}
where 
\begin{align}
    \mu_1 = \underset{x \in [\theta_0, \theta]}{\text{min}} \left( x + w (x) - b \right).
\end{align}
Denote the right hand side of equation as LHS.
\begin{align}
    | \text{LHS} | & \le \sint_{-\infty}^{\theta_0} \, ( x + w(x) - b_0)  - ( x + w(x) - b) \, dx \;\le\; \mu_2 |b-b_0|,
\end{align}
where
\begin{align}
    \mu_2 = \underset{b \in [b_0, b]}{\text{max}} \left( \theta_0 - \psi_l(b) \right).
\end{align}
Therefore, we have
\begin{align}
    | \theta - \theta_0 |\le \frac{\mu_2}{\mu_1} |b-b_0|.
\end{align}
\end{proof}
Lemma contains the theorem. Hence the proof is done.
\end{proof}

\begin{theorem} \label{theorem:gmax4overdrivedelay}
Given a constant off\nobreakdash-time current control loop with comparator overdrive delay, if the input is a ramp with slope $m_1$ and interference function $w(t)$,
the maximum comparator overdrive delay $t_{od}^{\text{max}}$ is
\begin{align}\label{eqn:tau_constraint}
    t_{od}^{\text{max}} = \frac{A_{ub}}{m_1} + \sqrt{\left(\frac{A_{ub}}{m_1}\right)^2+\frac{2}{m_1}\left(V_{\text{th}}\tau + B \right)},
\end{align}
where
\begin{align}
  \tau = \frac{C_{\text{eff}}}{G}, \quad B =  \bigintssss_{-\infty}^{+\infty} \Bigg| \frac{W(\omega)}{\omega \cdot \pi }  \Bigg| \,d\omega \nonumber.
\end{align}
\end{theorem}
\begin{proof}
We omitted the discrete-time notation $[n]$ for ease of derivation. The current error \mbox{$i_e \triangleq i_c - i_p$} is defined as the difference between the peak current command $i_c$ and the actual peak current $i_p$.

We define the \emph{overdrive delay} as
\begin{equation}
t_{od} \triangleq t_{\text{on}} - t_c.
\end{equation}

Given that the integrator does not saturate for $t > t_c$, the crossing time $t_c$ and the turn-on time $t_{\text{on}}$ satisfy
\begin{equation}
\frac{1}{C_{\text{eff}}}
\int_{t_c}^{t_{\text{on}}}
g\!\left(
i_e + m_1 (t - t_{\text{on}}) + w(t)
\right)\, dt
\;=\;
V_{\text{th}} .
\end{equation}

The transconductance $g$ is bounded from above by $G$. $G$ provides the shortest overdrive delay. Therefore we do the worst-case analysis by substituting $g$ by $G$. Then the worst-case comparator time constant is defined as $\tau_c \triangleq C_{\text{eff}}/G$
\begin{align} 
    \int_{t_c}^{t_{\text{on}}} \left( i_e + m_1 (t - t_{\text{on}}) + w(t) \right)\,dt \;=\; V_{\text{th}}\tau_c.
\end{align}

Given that integrator does not saturate after $t>t_c$, the crossing time $t_c$ and on-time $t_{\text{on}}$ satisfy
\begin{align}
    \frac{1}{C_{\text{eff}}}\int_{t_c}^{t_{\text{on}}} g\left(i_e + m_1 (t-t_{\text{on}}) + w(t) \right)\,dt = V_{\text{th}}.
\end{align}
The transconductance $g$ is bounded from above by $G$. $G$ provides the shortest overdrive delay. Therefore we do the worst-case analysis by substituting $g$ by $G$. Then the worst-case comparator time constant is defined as $\tau_c \triangleq C_{\text{eff}}/G$
\begin{align} \label{eqn:w_vth_g_identity}
    \int_{t_c}^{t_{\text{on}}} \left( i_e + m_1 (t - t_{\text{on}}) + w(t) \right)\,dt = V_{\text{th}}\tau_c.
\end{align}
Equation (\ref{eqn:w_vth_g_identity}) can be equivalently written as 
\begin{align} \label{IT8}
\frac{1}{2}m_1 (t_{\text{on}} - t_c)^2 - w(t_c)(t_{\text{on}} - t_c) + \int_{t_c}^{t_{\text{on}}} w(t)\; dt = V_{\text{th}}\tau_c.
\end{align}
If $W(\omega)/\omega$ is absolute integrable
\begin{align}
    \Big| \bigintssss_{-\infty}^t f(\tau)\,d\tau \Big| & = \frac{1}{2\pi} \Bigg|\bigintssss_{-\infty}^{+\infty} 
    \left(\frac{W(\omega)}{j\omega} + \pi W(0) \delta(\omega) \right) e^{j\omega t}\,d\omega \Bigg| \le \frac{1}{2\pi} \bigintssss_{-\infty}^{+\infty} \bigg | \frac{W(\omega)}{\omega} \bigg | \,d\omega.
\end{align}
We can bound the integral term in (\ref{IT8}) as
\begin{align} \label{eqn:boundf}
     \bigg | \bigintssss_{t_c}^{t_{\text{on}}} w(t)\; dt \bigg| \le 
      \bigg | \bigintssss_{-\infty}^{t_c} w(t)\; dt \bigg| +
      \bigg | \bigintssss_{-\infty}^{t_{\text{on}}} w(t)\; dt \bigg| 
     \le
     \frac{1}{\pi}\bigintssss_{-\infty}^{+\infty} \bigg | \frac{W(\omega)}{\omega}  \bigg| \,d\omega.
\end{align}
From (\ref{IT8}),
\begin{align} \label{eqn:doubletd}
t_{od} = \frac{w(t_c)}{m_1} \pm \sqrt{\left(\frac{w(t_c)}{m_1}\right)^2 + \frac{2}{m_1}\left(k-\int_{t_c}^{t_{\text{on}}} w(t)\; dt \right)}.
\end{align}
By assuming $V_{\text{th}}\tau_c > B$ \footnote{The justification for this assumption is for stability guarantee, which can be found in Theorem\,\ref{theorem:stability1}}, we narrow down two solutions in (\ref{eqn:doubletd}) to one feasible solution
\begin{align}
t_{od} = \frac{w(t_c)}{m_1} + \sqrt{\left(\frac{w(t_c)}{m_1}\right)^2 + \frac{2}{m_1}\left(V_{\text{th}}\tau_c-\int_{t_c}^{t_{\text{on}}} w(t)\; dt \right)}.
\end{align}
From (\ref{eqn:boundf}), the variable delay is bounded from above by $t_u$ and from below by $t_l$
\begin{align}
t_u \triangleq \frac{w(t_c)}{m_1} + \sqrt{\left(\frac{w(t_c)}{m_1}\right)^2+\frac{2}{m_1}\left(V_{\text{th}}\tau_c + B \right)}, \label{eqn:tu}\\
t_l \triangleq \frac{w(t_c)}{m_1} + \sqrt{\left(\frac{w(t_c)}{m_1}\right)^2+\frac{2}{m_1}\left(V_{\text{th}}\tau_c - B \right)} \label{eqn:tl}.
\end{align}

The longest delay can be obtained from (\ref{eqn:tu}) as
\begin{align}
    t^{\text{max}}_{od} = \frac{A}{m_1} + \sqrt{\left(\frac{A}{m_1}\right)^2+\frac{2}{m_1}\left(V_{\text{th}}\tau_c + B \right)}.
\end{align}
\end{proof}

\begin{theorem} \label{theorem:stability1}
Given a constant off\nobreakdash-time current control loop with comparator overdrive delay, if the input is a ramp with slope $m_1$ and interference function $w(t)$,
The condition to guarantee a globally asymptotically stable dynamical mapping is
\begin{align}\label{eqn:vtrigtau_constraint}
  V_{\text{trig}} \tau \ge \frac{4A_{ub}^2}{m_1} + B,
\end{align} 
\end{theorem}

\begin{proof}
\subsection{Sector Boundedness of Nonlinearity $\psi$}
$\psi$ represents the mapping from $\tilde{t}_{\text{on}}$ to $\tilde{i}_e$.
Although $\psi$ is not necessarily differentiable everywhere, we can still prove that $\psi$ is sector bounded by \mbox{$[-2A_{ub}/t_{od}^{\text{min}},+\infty)$} or equivalently in math as
\begin{align}
\label{eqn:sc_bd_psi_opd}
  \left( \psi(\tilde{t}_{\text{{on}}}) + \frac{A_{ub}}{t_{od}^{\text{min}}}\tilde{t}_{\text{{on}}} 
  \right) \tilde{t}_{\text{{on}}} \;>\; 0,
\end{align}
where $t_{od}^{\text{min}}$ is the minimum comparator overdrive delay.
\vspace{+6pt}

\noindent(1)\quad We perturb $i_e$ by $d\,i^{+}_e$, which is a positive infinitesimal. Because of the continuity property of the static mapping, the resulting increase on the on-time $d\,t^{+}_{\text{on}}$ is also a positive infinitesimal. However, the variation on crossing time $\Delta\,t^{+}_c$ might be a positive large number
\begin{align}
     \int^{t_{\text{on}}+dt^{+}_{\text{on}}}_{t_c + \Delta t_c^{+}} -(i_e + di_e^{+}) + m_1 (t - t_{\text{on}} - dt_{\text{on}}^{+}) + w(t)\, dt = V_{\text{th}}\tau_c. 
\end{align}
\begin{align}
    & \int_{t_c + \Delta t_c^{+}}^{t_{\text{on}}} -i_e + m_1(t - t_{\text{on}}) + w(t)\,dt + \int_{t_c}^{t_c + \Delta t_c^{+}}  -i_e + m_1(t - t_{\text{on}}) + w(t)\,dt \;=\; V_{\text{th}}\tau_c.
\end{align}
Because $-i_e + m_1(t - t_{\text{on}}) + w(t)$ is positive for all $t \ge t_c$, we have
\begin{align} \label{eqn:pos_perturb_inequ}
    &\int^{t_{\text{on}}+dt^{+}_{\text{on}}}_{t_c + \Delta t_c^{+}} -(i_e + di_e^{+}) + m_1 (t - t_{\text{on}} - dt_{\text{on}}^{+}) + w(t)\, dt\; \ge \;
    \int_{t_c + \Delta t_c^{+}}^{t_{\text{on}}} -i_e + m_1(t - t_{\text{on}}) + w(t)\,dt.
\end{align}
Equation (\ref{eqn:pos_perturb_inequ}) can be algebraically transformed as 
\begin{align}
\label{eqn:pos_perturb_inequ_2}
    & -(i_e + di_e^{+})(t_{\text{on}} + dt^{+}_{\text{on}} - t_c - \Delta t_c^{+}) + \int_{t_c + \Delta t_c^{+}}^{t_{\text{on}}+dt_{\text{on}}^{+}} w(t) dt 
    -\frac{m_1}{2}(t_c + \Delta t^{+}_c-t_{\text{on}} - d t_{\text{on}}^{+})^2 \nonumber \\
    & \ge -i_e(t_{\text{on}}-t_c-\Delta t_c^{+}) -\frac{m_1}{2}(t_c + \Delta t_c^{+}-t_{\text{on}})^2+\int_{t_c+\Delta t_c^{+}}^{t_{\text{on}}} w(t) dt.
\end{align}
Considering $di_e^{+}$ and $dt_{\text{on}}^{+}$ are infinitesimals,  (\ref{eqn:pos_perturb_inequ_2}) can be rearranged as 
\begin{align} \label{eqn:pos_perturb_inequ_3}
    & - i_e dt^{+}_{\text{on}} 
    - d i^{+}_{e}(t_{\text{on}} - t_c - \Delta t_c^{+}) + d t_{\text{on}}^{+}w(t_{\text{on}}) 
     + m_1 (t_c + \Delta t_c^{+}-t_{\text{on}}) \ge 0.
\end{align}
We substitute (\ref{eqn:pos_perturb_inequ_4}) into (\ref{eqn:pos_perturb_inequ_3}). In this way, (\ref{eqn:pos_perturb_inequ_3}) is simplified into (\ref{eqn:pos_perturb_inequ_5})
\begin{align} \label{eqn:pos_perturb_inequ_4}
    -(i_e + di_e^{+}) + m_1(t_c + \Delta^{+}t_c - t_{\text{on}}) +  w(t_c + \Delta t_c^{+}) \;=\; 0,
\end{align}
\begin{align} \label{eqn:pos_perturb_inequ_5}
    (t_{\text{on}} - t_c - \Delta t_c^{+}) d i_e^{+} \;\le\; \left( w(t_{\text{on}}) - w(t_c + \Delta t^{+}_{c})\right) d t_{\text{on}}^{+}.
\end{align}
We observe that $ (t_{\text{on}} - t_c - \Delta t_c^{+})$, and $d i_e^{+}$, $d t_{\text{on}}^{+}$ are all positive. Therefore,
the right derivative of the function from $t_{\text{on}}$ to $i_e$ is bounded by
\begin{align} \label{eqn:right_derv_bd}
    0 \;\le\; \frac{d i_e^{+}}{d t_{\text{on}}^{+}} 
    \;\le\; \frac{w(t_{\text{on}}) - w(t_c + \Delta t_c^{+})}{t_{\text{on}}-(t_c + \Delta t^{+}_{c})} \;\le\; \frac{2A_{ub}}{t_{od}^{\text{min}}},
\end{align}
where $t_{od}^{\text{min}}$ is the minimum comparator overdrive delay.

\vspace{+6pt}
\noindent(2)\quad We perturb $i_e$ by $d i^{-}_e$ which is a negative infinitesimal. A similar inequality is obtained as
\begin{align}
     (t_{\text{on}} - t_c - \Delta t_c^{-}) d i_e^{-} \le \left( w(t_{\text{on}}) - w(t_c + \Delta t^{-}_{c})\right) d t_{\text{on}}^{-}.
\end{align}
The left derivative of the function from $t_{\text{on}}$ to $i_e$ is bounded by
\begin{align} \label{eqn:left_derv_bd}
    \frac{d i_e^{-}}{d t_{\text{on}}^{-}} 
     \ge \frac{w(t_{\text{on}}) - w(t_c + \Delta t_c^{-})}{t_{\text{on}}-(t_c + \Delta t^{-}_{c})} \ge -\frac{2A_{ub}}{t_{od}^{\text{min}}}.
\end{align}
\vspace{+6pt}
\noindent(3)\quad Therefore, $\psi$ is sector bounded by 
\begin{align}
  \left( \psi(\tilde{t}_{\text{{on}}}) + \frac{A_{ub}}{t_{od}^{\text{min}}}\tilde{t}_{\text{{on}}} 
  \right) \tilde{t}_{\text{{on}}} > 0.
\end{align}
\subsection{Bounds of the Comparator Overdrive Delay}
From the circle criterion,
a sufficient and necessary condition for the stability of current control loop is
\begin{align} \label{eqn:abs_stab_m1}
    \text{min} \Bigg \{ \frac{d i_e^{+}}{d t_{\text{on}}^{+}}, \frac{d i_e^{-}}{d t_{\text{on}}^{-}} \Bigg \}  > -\frac{m_1}{2}.
\end{align}
From (\ref{eqn:right_derv_bd}) 
\begin{align}
    \frac{d i_e^{+}}{d t_{\text{on}}^{+}} \ge 0.
\end{align}
From (\ref{eqn:left_derv_bd}) 
\begin{align}
\frac{d i_e^{-}}{d t_{\text{on}}^{-}} \ge \frac{w(t_{\text{on}})-w(t_{\text{c}})}{t_{\text{on}} - t_{c}}.
\end{align}
We can prove (\ref{eqn:abs_stab_m1}) if we can show the following:
\begin{align}
    \frac{w(t_{\text{on}})-w(t_{\text{c}})}{t_{\text{on}} - t_{c}} > -\frac{m_1}{2}.
\end{align}

We prove that $V_{\text{th}} \tau_c$ can guarantee the stability by using (\ref{eqn:abs_stab_m1}). The lower sector bound of the nonlinear function in the feedback path is 
\begin{align}
    & \frac{w(t_{\text{on}}) - w(t_c)}{t_{\text{on}} - t_c} \ge \, \frac{-A - w(t_c)}{t_{\text{on}} - t_c} \ge \frac{-A- w(t_c)}{t_l}=
    \label{eq:quo}
\end{align}
\begin{align}
    & \frac{-A -w(t_c)}{\frac{w(t_c)}{m_1} + \sqrt{\left(\frac{w(t_c)}{m_1}\right)^2+\frac{2}{m_1}\left(V_{\text{th}}\tau_c - B \right)}} =\\
    & \, - \frac{1+\frac{w(t_c)}{A}}{\frac{w(t_c)}{A} + \sqrt{\left(\frac{w(t_c)}{A}\right)^2 + \frac{2m_1}{A^2}\left(V_{\text{th}}\tau_c - B
    \right)}}m_1 \ge \label{eqn:stabproofbeforemid}\\
    & \, - \frac{1+\frac{w(t_c)}{A}}{\frac{w(t_c)}{A} + \sqrt{\left(\frac{w(t_c)}{A}\right)^2 + 8}}m_1 \ge  \label{eqn:stabproofbeforeend}\\
    & \, -\frac{m_1}{2}. \label{eqn:stabproofend}
\end{align}
The derivation from (\ref{eqn:stabproofbeforeend}) to (\ref{eqn:stabproofend}) is because of the following auxiliary function. In our case, we demonstrate the monotonicity of $y(x)$ by defining $x = \frac{w(t_c)}{A}$ restricting it to $x \in [-1,1]$ and $\mu>0$.
\begin{align}
    y \triangleq \frac{-1-x}{x+\sqrt{x^2+\mu}}.
\end{align}
Differentiate $y(x)$ and simplify:
\begin{align}
y'(x)
&= \frac{x - \sqrt{x^2+\mu} + 1}{\,\mu + x^2 + x\sqrt{x^2+\mu}\,}.
\label{eq:ydash}
\end{align}
The denominator is strictly positive for all $x\in\mathbb{R}$:
\begin{align}
\mu + x^2 + x\sqrt{x^2+\mu}
= \sqrt{x^2+\mu}\,\bigl(\sqrt{x^2+\mu}+x\bigr) \;>\; 0.
\label{eq:denpos}
\end{align}
Hence the sign of $y'(x)$ equals the sign of the numerator
\begin{equation}
N(x)\;=\;x-\sqrt{x^2+\mu}+1.
\end{equation}
To make $y$ monotone \emph{decreasing} on $[-1,1]$, it suffices to require $N(x)\le 0$ for all $x\in[-1,1]$:
\begin{equation}
    x-\sqrt{x^2+\mu}+1 \;\le\; 0
\;\;\Longleftrightarrow\;\;
\sqrt{x^2+\mu} \;\ge\; x+1.
\end{equation}
On $[-1,1]$ we have $x+1\in[0,2]$, so both sides are nonnegative and we may square:
\begin{equation}
x^2+\mu \;\ge\; (x+1)^2 \;=\; x^2+2x+1
\;\;\Longleftrightarrow\;\;
\mu \;\ge\; 2x+1.
\end{equation}
Thus, the condition must hold for the worst (largest) right-hand side over $x\in[-1,1]$:
\begin{equation}
\mu \;\ge\; \max_{x\in[-1,1]}(2x+1) \;=\; 2\cdot 1 + 1 \;=\; 3.
\end{equation}
Therefore, \eqref{eqn:stabproofbeforeend} is monotonically decreasing on the domain $x \in [-1,1]$ provided
\begin{equation}
    \mu \;\ge\; 3,
    \qquad
    \mu \;\triangleq\; \frac{2m_1}{A^2}\,\bigl(V_{\text{th}}\tau_c - B\bigr).
    \label{eq:mu_def_mono}
\end{equation}
Equivalently, 
\begin{equation}
    u \;\ge\; \frac{3A^2}{2m_1}
    \quad\Longleftrightarrow\quad
    V_{\text{th}}\tau_c \;\ge\; \frac{3A^2}{2m_1} \,+\,B.
    \label{eq:Vth_tau_threshold}
\end{equation}
Furthermore, for $\mu = 8$, we obtain $y_{\text{min}} = y\big|_{x = 1} = -0.5$, 
which guarantees the stability of the system. Hence, the stability condition can be expressed as
\begin{equation}
    V_{\text{th}}\,\tau_{c} \;\ge\; \frac{4A^{2}}{m_{1}} \;+\; B.
    \label{eq:cond}
\end{equation}
\end{proof}
\begin{theorem}
Consider a constant off-time current control loop with the overdrive delay. The input is a ramp signal with slope $m_1 (> 0)$ and an interference function $w(t)$.  
Given an interference amplitude upper bound $A$ and an integral term $B$, the linearized feedback gain $\psi$ is bounded within the range $[\psi_{\min},\,\psi_{\max}]$, where
\begin{equation}
\psi_{\min}
= \frac{-2 m_1}{
1 + \sqrt{1 + \frac{2m_1}{A^{2}}
\left(V_{\text{th}}\tau_c - B\right)}
},
\qquad
\psi_{\max}
= \frac{2 m_1}{
-1 + \sqrt{1 + \frac{2m_1}{A^{2}}
\left(V_{\text{th}}\tau_c - B\right)}
}.
\end{equation}

\end{theorem}

\begin{proof}
We define
\begin{equation}
    \psi \triangleq \frac{w(t_{\text{on}}) - w(t_c)}{t_{\text{on}} - t_c}.
\end{equation}
From~\eqref{eqn:stabproofbeforemid}, $\psi$ is bounded as
\begin{equation}
\psi \in 
\Bigg[
- \frac{\bigl(1+\frac{w(t_c)}{A}\bigr)m_1}
{\frac{w(t_c)}{A} + \sqrt{\left(\frac{w(t_c)}{A}\right)^2 
+ \frac{2m_1}{A^2}\!\left(V_{\text{th}}\tau_c - B\right)}},\;
- \frac{\bigl(1+\frac{w(t_c)}{A}\bigr)m_1}
{\frac{w(t_c)}{A} + \sqrt{\left(\frac{w(t_c)}{A}\right)^2 
+ \frac{2m_1}{A^2}\!\left(V_{\text{th}}\tau_c + B\right)}}
\Bigg].
\end{equation}

Given that the magnitude of the interference function $w(t)$ is bounded by $A$, we have
\begin{equation}
    |w(t)| \le A, \qquad -A \le w(t) \le A .
\end{equation}
Denoting the minimum and maximum values of $\psi$ by $\psi_{\min}$ and $\psi_{\max}$, and noting that the extrema\footnote{The denominator is strictly positive, as stated in Theorem~3.} occur at $w(t_c)=\pm A$, it follows that
\begin{equation}
\psi_{\min}
= \min\!\left[
\frac{-2m_1}{1 + \sqrt{1 + \frac{2m_1}{A^2}\!\left(V_{\text{th}}\tau_c - B\right)}},
\; 0
\right]
= \frac{-2m_1}{1 + \sqrt{1 + \frac{2m_1}{A^2}\!\left(V_{\text{th}}\tau_c - B\right)}},
\end{equation}
and
\begin{equation}
\psi_{\max}
= \max\!\left[
0,\;
\frac{2m_1}{1 + \sqrt{1 + \frac{2m_1}{A^2}\!\left(V_{\text{th}}\tau_c - B\right)}}
\right]
= \frac{2m_1}{-1 + \sqrt{1 + \frac{2m_1}{A^2}\!\left(V_{\text{th}}\tau_c - B\right)}} .
\end{equation}

\end{proof}

\section{Analytical Approximation and Numerical Evaluation of the Integral Term $B$}
We first approximate the interference effect on the comparator overdrive delay
by assuming a single-tone interference model
$w(t) = A\sin(\omega t + \phi)$.
This analytical form allows us to evaluate the integral term explicitly
and approximate the resulting $B$ defined in~\eqref{eqn:boundf}.
Building upon this simplified case, we then extend the analysis to the
realistic scenario where $w(t)$ does not admit a closed-form analytical
expression. In this general case, we propose a tighter frequency-domain
bound on $B$ that remains valid for arbitrary interference waveforms.

\subsection{Analytical Approximation}
The integral in~\eqref{eqn:w_vth_g_identity} can be expressed analytically as
\begin{equation}
\label{eq:f_tphi}
\frac{m_1}{2}t^2 - \chi(t,\phi)
= V_{\text{th}}\tau_c,
\end{equation}
where the interference-related integration term $\chi(t,\phi)$ is defined as
\begin{equation}
\label{eq:second_term}
\chi(t,\phi)
= -\frac{A}{\omega}\big[\cos(\omega t+\phi)-\cos(\phi)\big].
\end{equation}
To determine the worst-case contribution, we seek the maximum magnitude of $\chi(t,\phi)$ over all possible phases $\phi$:
\begin{equation}
\label{eq:T2max_def}
|\chi|_{\max}(t)
= \max_{\phi}
\left|
\frac{A}{\omega}
\big[\cos(\omega t+\phi)-\cos(\phi)\big]
\right|.
\end{equation}
Using the trigonometric identity
\[
\cos a - \cos b = -2\sin\!\Bigl(\tfrac{a+b}{2}\Bigr)\sin\!\Bigl(\tfrac{a-b}{2}\Bigr),
\]
we obtain
\begin{equation}
\cos(\omega t+\phi) - \cos(\phi)
= -2\sin\!\Bigl(\phi+\tfrac{\omega t}{2}\Bigr)
   \sin\!\Bigl(\tfrac{\omega t}{2}\Bigr).
\end{equation}
For a fixed $t$, $\sin(\omega t/2)$ is constant, and $\sin(\phi+\omega t/2)$
can vary between $-1$ and $1$. Therefore, the worst-case magnitude of $\chi(t,\phi)$ is
\begin{equation}
\label{eq:T2max_final}
|\chi|_{\max}(t)
= \frac{2A}{\omega}
  \big|\sin\!\bigl(\tfrac{\omega t}{2}\bigr)\big|.
\end{equation}
For small integration durations such that $\omega t \ll 1$, the sine function can be linearized as $\sin(\tfrac{\omega t}{2}) \approx \tfrac{\omega t}{2}$. Substituting this approximation yields
\begin{equation}
\label{eq:chi_max_smallangle}
|\chi|_{\max}(t)
\approx
A t.
\end{equation}
Substituting the worst-case interference amplitude into~\eqref{eq:f_tphi} gives
\begin{equation}
\label{eq:f_with_chimax}
\frac{m_1}{2}t^2
- \frac{2A}{\omega}\sin\!\bigl(\tfrac{\omega t}{2}\bigr)
= V_{\text{th}}\tau_c.
\end{equation}
Applying the same small-angle approximation $\sin(\tfrac{\omega t}{2}) \approx \tfrac{\omega t}{2}$ leads to
\begin{equation}
\label{eq:t_quadratic}
\frac{m_1}{2}t^2 - A t = V_{\text{th}}\tau_c.
\end{equation}
Solving~\eqref{eq:t_quadratic} for $t$ gives
\begin{equation}
\label{eq:t_solution}
t =
\frac{
A + \sqrt{A^2 + 2m_1 V_{\text{th}}\tau_c}
}{
m_1
}.
\end{equation}
The corresponding $B$ value can be expressed as
\begin{equation}
\label{eq:B_expression}
B \approx
A \cdot
\frac{
A + \sqrt{A^2 + 2m_1 V_{\text{th}}\tau_c}
}{
m_1
}.
\end{equation}
Note that this expression does not represent an absolute bound as in~\eqref{eqn:boundf}
Due to the small-angle approximation, however, it typically yields a value much smaller
than the conservative bound proposed in~\eqref{eqn:boundf}.

\subsection{Exact Numerical Formulation for $B$}
When the interference signal $w(t)$ consists of multiple sinusoidal components with independent phases, it can be expressed as
\begin{equation}
w(t) = \sum_{k=1}^{K} A_k \sin(\omega_k t + \phi_k).
\label{eq:w_multi_sin}
\end{equation}
The integral of $w(t)$ over an interval $[0,t]$ is
\begin{equation}
\int_{0}^{t} w(\tau)\,d\tau
= \sum_{k=1}^{K} 
    \frac{2A_k}{\omega_k}
    \sin\!\Bigl(\tfrac{\omega_k t}{2}\Bigr)
    \sin\!\Bigl(\phi_k + \tfrac{\omega_k t}{2}\Bigr).
\label{eq:multi_sin_integral}
\end{equation}
Since each phase $\phi_k$ can vary independently, the maximum possible magnitude of the total integral occurs when all sinusoidal terms add constructively. The resulting worst-case bound is therefore
\begin{equation}
\label{eq:int_bound_independent}
\Biggl|\int_{0}^{t} w(\tau)\,d\tau\Biggr|
\;\le\;
\sum_{k=1}^{K}
    \frac{2A_k}{\omega_k}
    \biggl|\sin\!\Bigl(\tfrac{\omega_k t}{2}\Bigr)\biggr|.
\end{equation}
Hence, \( B(t) \) is introduced to indicate the variation of \( t \) between \( 0 \) and \( t_{\text{od}}^{\text{max}} \), and can be expressed as
\begin{equation}
\label{eq:B_tight_bound}
B(t)
=
\int_{-\infty}^{+\infty}
\biggl|
\frac{W(\omega)}{\omega\cdot\pi}
\biggr|
\cdot
\biggl|
\sin\!\Bigl(
\tfrac{
\omega t
}{2}
\Bigr)
\biggr|\,d\omega
\end{equation} 
This formulation introduces a time-dependent modulation term that captures the finite integration window of the comparator delay. The corresponding value of $B$ can be obtained by applying a numerical solver to the delay equation in \eqref{eqn:doubletd}.

\section{A Case Study on the Application of High-Speed Current-Mode Control in DC--DC Converters}
\begin{figure}[t]
\centering
\includegraphics[width=\textwidth]{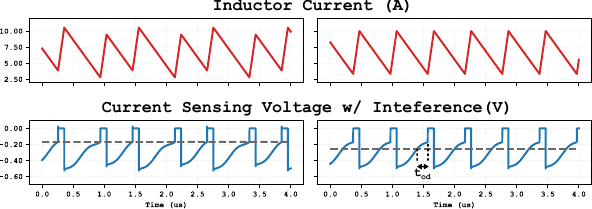}
\caption{Simulated waveforms of the inductor current and the current-sensing voltage under interference, obtained from a bottom current-sensing configuration. The left plot shows an ideal comparator with \(V_{\text{th}}\tau_c = 0\), exhibiting subharmonic inductor current oscillation and instability. The right plot shows a comparator with \(V_{\text{th}}\tau_c = 9.24~\text{V}\!\cdot\!\text{ns}\), where the instability is eliminated. The gray dashed line indicates the current reference, and the overdrive voltage corresponds to the region above this threshold.}
\end{figure}
Here, we present a practical case study for designing and selecting the threshold--time constant product \(V_{\text{th}}\tau_c\) of the comparator in a high-speed DC--DC converter~\cite{qian_analog_2025} to guarantee stability even under severe interference. The trigger integral is given by
\begin{equation}
\label{eq:int_condition_full}
\int_{0}^{t_{od}}
\!\!\Big(G_s m t + A\cos(\omega_1 t + \phi)\sin(\omega_2 t + \phi)\Big)\, dt
\;=\;
V_{\text{th}}\tau_c,
\end{equation}
where the term inside the integral represents the sensed current voltage. Moreover, \(A\) denotes the interference amplitude, \(\omega_1\) and \(\omega_2\) are the respective angular frequencies, and \(\phi\) is the phase angle. The interference is modeled as the product of cosine and sine components. The parameter \(G_s\) denotes the current-sensing gain, typically realized through a small sense resistor, and \(m\) represents the inductor current slope, either rising or falling, depending on the constant on-time or constant off-time configuration.
In this case study, let \(m\) denote the \emph{falling} slope of the inductor current, given by \(m = V_{\text{out}}/L\). With \(L = 150~\text{nH}\), \(V_{\text{out}} = 2~\text{V}\), and a sensing gain \(G_s = 50~\text{mV/A} = 0.05~\text{V/A}\), we obtain
\begin{equation}
m \;=\; \frac{V_{\text{out}}}{L}
= \frac{2}{150\times 10^{-9}}
= 1.333\times 10^{7}\ \text{A/s},\qquad
mG_s \;=\; 6.667\times 10^{5}\ \text{V/s}.
\end{equation}
The design must satisfy the following condition to ensure stability:
\begin{equation}
\label{eq:vth_tau_condition_num}
V_{\text{th}}\tau_{c}
\;\ge\;
\frac{4A^{2}}{m\,G_s}
\;+\;
B,
\qquad
 B
=
\int_{-\infty}^{+\infty}
\biggl|
\frac{W(\omega)}{\omega \cdot \pi}
\biggr|
\,d\omega.
\end{equation}
The maximum allowable overdrive delay can then be expressed as
\begin{equation}
t^{\text{max}}_{od} 
= 
\frac{A}{mG_s} 
+ 
\sqrt{
\left(\frac{A}{mG_s}\right)^2
+
\frac{2}{mG_s}
\left(V_{\text{th}}\tau_c + B \right)
}.
\end{equation}
For simplicity, let \(\omega_1 = \omega_2 = \omega\) in~\eqref{eq:int_condition_full}. Taking
\begin{equation}
A = 30~\text{mV} = 0.03~\text{V}, 
\qquad
\omega = 2\pi \times 1.125 \times 10^{6}~\text{rad/s},
\end{equation}
and solving~\eqref{eq:vth_tau_condition_num} yields the minimum feasible threshold--time constant product, along with the corresponding interference bound and overdrive delay:
\begin{equation}
V_{\text{th}}\tau_c = 9.24~\text{V}\!\cdot\!\text{ns},
\qquad
B = \frac{A}{\omega} = 4.24~\text{V}\!\cdot\!\text{ns},
\qquad
t_{\text{od}}
= 
\sqrt{\frac{2V_{\text{th}}\tau_c}{mG_s}}
= 166.5~\text{ns}.
\end{equation}
With \(V_{\text{th}}\tau_c = 0\), the inductor current exhibits subharmonic oscillation. By introducing a comparator that satisfies the above criterion, the instability is fully suppressed. The resulting measured delay, \(t_{od} = 175.5~\text{ns}\), is also smaller than the theoretical maximum \(t_{od}^{\text{max}}\) (280\,ns).

\section{Conclusion}
This note establishes a rigorous analytical model for the comparator overdrive delay, leading to robust high-frequency current-mode control. It proves the continuity condition for the static current mapping under sensor interference, derives closed-form sector bounds on the nonlinearity in the dynamic mapping to enable formal stability certification, and obtains tight analytical upper and lower bounds on the comparator overdrive delay. In conjunction with the analytical and numerical characterization of the interference term \(B\) and the accompanying case study, these results provide practical design guidelines to ensure the stability of current-mode converters operating under sensor interference.

\bibliographystyle{ieeetr}
\bibliography{main}

@misc{compoverdelayadi,
title = {{Parameters that Affect Comparator Propagation Delay Measurements | Analog Devices}},
url = {https://www.analog.com/en/technical-articles/parameters-that-affect-comparator-propagation-delay-measurements.html},
urldate = {2022-06-14}
}

@inproceedings{Cui2018a,
address = {Padova},
author = {Cui, Xiaofan and Avestruz, Al-Thaddeus},
booktitle = {2018 IEEE 19th Workshop on Control and Modeling for Power Electronics (COMPEL)},
doi = {10.1109/COMPEL.2018.8460055},
isbn = {VO -},
pages = {1--8},
title = {{A New Framework for Cycle-by-Cycle Digital Control of Megahertz-Range Variable Frequency Buck Converters}},
year = {2018}
}

@article{Fernandes2016,
author = {Fernandes, Ryan and Trescases, Olivier},
doi = {10.1109/TPEL.2015.2499194},
issn = {08858993},
journal = {IEEE Transactions on Power Electronics},
month = {aug},
number = {8},
pages = {5694--5708},
publisher = {Institute of Electrical and Electronics Engineers Inc.},
title = {{A Multimode 1-MHz PFC Front End with Digital Peak Current Modulation}},
volume = {31},
year = {2016}
}

@article{Theory,
author = {Theory, Operational Amplifiers},
title = {{MITRES_6-010S13_comchaptrs}}
}

@inproceedings{qian_analog_2025,
	address = {Anaheim, CA, USA},
	title = {Analog {Comparator} {Modeling} and {Selection} for {High}-{Speed} {Constant} {On}-{Time} {Controlled} {DC}-{DC} {Converters}},
	copyright = {https://doi.org/10.15223/policy-029},
	isbn = {979-8-3315-2214-8},
	url = {https://ieeexplore.ieee.org/document/11097936/},
	doi = {10.1109/ITEC63604.2025.11097936},
	language = {en},
	urldate = {2025-08-11},
	booktitle = {2025 {IEEE}/{AIAA} {Transportation} {Electrification} {Conference} and {Electric} {Aircraft} {Technologies} {Symposium} ({ITEC}+{EATS})},
	publisher = {IEEE},
	author = {Qian, Guanyu and Cui, Xiaofan},
	month = jun,
	year = {2025},
	pages = {1--8},
}
\end{document}